\documentclass[reqno]{amsart}

\usepackage{amsmath,amssymb}

\newtheorem{theorem}{Theorem}[section]
\newtheorem{proposition}[theorem]{Proposition}
\newtheorem{lemma}[theorem]{Lemma}

\newtheorem{definition}[theorem]{Definition}
\newtheorem{remark}[theorem]{Remark}

\numberwithin{equation}{section}

\DeclareMathOperator{\supp}{supp}

\renewcommand\L{\mathrm{L}}
\newcommand\R{\mathbb R}

\newcommand\Z{\mathbb Z}

\newcommand{\cC}{\mathcal{C}}

\newcommand{\cA}{\mathcal{A}}

\newcommand{\cK}{\mathcal{K}}
\newcommand{\cM}{\mathcal{M}}

\newcommand{\bom}{\boldsymbol{\omega}}
\newcommand{\om}{\omega}

\newcommand\e{\mathrm{e}}

\newcommand\eps{\varepsilon}

\newcommand\La{\Lambda}

\renewcommand\P{\mathbb P}

\newcommand{\abs}[1]{\left\lvert #1 \right\rvert}

\newcommand{\set}[1]{\left\{ #1 \right\}}
\newcommand{\pa}[1]{\left( #1 \right)}

\newcommand\beq{\begin{equation}}
\newcommand\eeq{\end{equation}}

\begin{document}

\title[Localization for Cantor-Anderson]
{Localization for a continuum Cantor-Anderson  Hamiltonian}

\author{Fran\c cois Germinet}
\address{ Universit\'e de Cergy-Pontoise,
D\'epartement de Math\'ematiques,
Site de Saint-Martin,
2 avenue Adolphe Chauvin,
95302 Cergy-Pontoise cedex, France}
 \email{germinet@math.u-cergy.fr}

\author{Abel Klein}
\address{University of California, Irvine,
Department of Mathematics,
Irvine, CA 92697-3875,  USA}
 \email{aklein@uci.edu}

\thanks{A.K was  supported in part by NSF Grant DMS-0457474.}

\begin{abstract}
We prove localization at the bottom of the spectrum for a random Schr\"odinger operator in the continuum with a single-site potential probability distribution supported by a Cantor set of zero Lebesgue measure. This distribution is too singular to be treated by the usual methods. In particular,  an ``a priori" Wegner estimate is not available. To prove the result we perform a multiscale analysis following the work of Bourgain and Kenig for the Bernoulli-Anderson Hamiltonian, and obtain the required Wegner estimate scale by scale. To do so,  we generalize their argument based on Sperner's Lemma by resorting to the LYM inequality for multisets, and  combine it with the concept of scale dependent equivalent classes of configurations introduced by Germinet, Hislop and Klein for the study of Poisson Hamiltonians.
\end{abstract}

\dedicatory{Dedicated to Jean-Michel Combes on the occasion of his $\, 65^{\mathrm{th}}$ birthday.}

\maketitle

\section{Introduction and setting}
\label{sectintro}

Consider the continuum  Cantor-Anderson Hamiltonian
\begin{gather}\label{AndH}
H_{\bom} :=  -\Delta + V_{\bom} \quad \text{on} \quad \L^{2}(\R^{d}) , \\
\intertext{where the potential is given by}
  V_{\bom}(x):=  \sum_{\zeta \in \Z^d} \bom_\zeta \,u(x - \zeta), \label{AndV}
\end{gather}
where
\begin{itemize}
\item The single-site potential  $u$ is a  nonnegative, nonzero  $\mathrm{L}^{\infty}$-function on $\R^{d}$ with compact support, with 
 \begin{equation} \label{u}
u_{-}\chi_{\Lambda_{\delta_{-}}(0)}\le u \le u_{+}\chi_{\Lambda_{\delta_{+}}(0)}\quad \text{for some constants $u_{\pm}, \delta_{\pm}\in ]0,\infty[ $},
\end{equation}
$\Lambda_{L}(x)$ being  the 
box of side $L$ centered at $x \in \R^{d}$. 

\item $\bom=\{ \bom_\zeta\}_{\zeta \in \Z^d}$ is a family of independent, identically distributed random variables, with a the common probability distribution $\mu$ defined as the uniform measure on a Cantor set $\cK \subset[0,1]$ of Lebesgue measure zero ($\mu$ is constructed in Section~\ref{sectcantor}). We take the underlying probability space to be $(\Omega, \P)$
with $\Omega= \cK^{\Z^d}$, $\P=\mu^{\odot\Z^d}$, and $\{ \bom_\zeta\}_{\zeta \in \Z^d}$ the coordinate functions.
\end{itemize}

Jean-Michel Combes has made major contributions to the study of continuum  Anderson Hamiltonians,  operators of the form \eqref{AndH}-\eqref{AndV} where  the common probability distribution $\mu$ is absolutely continuous with a bounded density. Combes and Hislop  \cite{CH1} gave the first proof of Anderson localization for these random 
Schr\"odinger operators in the continuum.  Combes  and  his collaborators  \cite{CHM,CHN,CHKN,CHK1,CHKR} made important contributions in the understanding of the Wegner estimate, culminating in the recent beautiful paper \cite{CHK2}, which contains the optimal a priori Wegner estimate,  proving the long-sought Lipschitz continuity of the integrated density of states.  Combes also made important contributions to the study of  random Landau Hamiltonians \cite{CH2,CHKR}, which played an  important role in the dynamical delocalization result obtained for this model in \cite{GKS}. It is the pleasure of the authors to dedicate this paper to Jean-Michel Combes. This paper also deals with a Wegner estimate, but because of the singularity of the probability distribution, it requires alternative arguments than the ones developed by Jean-Michel Combes and his collaborators.

We now state our result. We use $\chi_x$ to denote the characteristic function of the unit cube $\Lambda_1(x)$.

\begin{theorem}\label{thmloc}
There exists $E_0>0$ such that $H_\omega$ exhibits Anderson localization as well as dynamical localization in the energy interval $[0,E_0]$. More precisely:
\begin{itemize}
\item (Anderson localization) There exists  $m>0$ such that, with probability one the operator $H_\omega$ has pure point spectrum  in $[0,E_0]$ with exponentially localized eigenfunctions with rate of decay $m$, i.e., if    $\phi$ is an eigenfunction of $H_\omega$ 
with eigenvalue $E \in[0,E_0]$ we have
\begin{equation}\label{expdecay}
 \|\chi_x \phi\| \le C_{\omega,\phi} \, e^{-m|x|}, \quad \text{for all $x \in \R^{d}$}.
\end{equation}

\item (Dynamical localization) For  all ${ s <  \frac3 8 d - }$  we have 
  \beq  {
\mathbb{E} \left\{ \sup_{t \in \mathbb{R}}
 \left\lVert {\langle} x {\rangle}^{\frac m 2}
e^{-itH_\omega} \chi_{[0,E_0]}(H_\omega) \chi_0 
 \right\rVert_2^{\frac {2 s} m} \right\}<
\infty \quad \text{for all} \; \;   m \ge 1}. \label{weakDL}
\eeq
\end{itemize}
 \end{theorem}

We thus obtain Anderson and dynamical localization with a single-site potential probability measure that is purely singular continuous with respect to Lesbegue measure. As shown in Section~\ref{sectcantor}, $\mu$ is $\log\log$-H\"older continuous which is too singular to enable us to use standard results for the Wegner estimate (such a priori Wegner estimates can be used to prove localization for single-site potential probability measures which are at least $\log$-H\"older continuous). As a consequence, the standard multiscale scale analysis \cite{FS,vDK,FK,GK1} cannot be exploited either. This is why the strong form of dynamical localization obtained in \cite{GK1} is replaced by the somewhat weaker form \eqref{weakDL} proved in \cite{GKnext}.

To prove the result we perform a multiscale analysis following  the work of Bourgain and Kenig for the Bernoulli-Anderson Hamiltonian \cite{BK}, and obtain the required Wegner estimate scale by scale. To do so, we generalize the argument based on Sperner's Lemma used in \cite{BK} by resorting to the LYM inequality for multisets (e.g., \cite{A}), and  combine it with the concept of scale dependent equivalent classes of configurations introduced by Germinet, Hislop and Klein for the study of Poisson Hamiltonians \cite{GHK1,GHK2}.

Theorem~\ref{thmloc} can be proven in great generality:  the only requirement on the  single-site potential probability measure $\mu$ is that $\set{0,\tau} \in \supp \mu \subset [0,\tau]$ for some $\tau >0$. The present note may be considered as an illustrative introduction to the general result proved in  \cite{GKnext} using the concentration bound of \cite{AGKW}.  Indeed, where here we restrict ourselves to a uniform measure and use the explicit hierarchical structure of its support (a particular Cantor set), new arguments had to be developed in \cite{AGKW,GKnext} in order to treat arbitrary measures. In particular, a Bernoulli decomposition of random variables is developed in \cite{AGKW}, which yields the concentration bound that extends the probabilistic consequences of the combinatorial Sperner's Lemma to general random variables.  This Bernoulli decomposition,   combined with the extension in \cite{GKnext} of the Bourgain-Kenig multiscale analysis to more general Bernoulli-Anderson Hamiltonians,  which incorporate an additional  background potential and  for which  the variances of the Bernoulli terms are uniformly positive, 
but not necessarily the same, yields pure point spectrum with probability one at the bottom of the  spectrum in the general case  \cite{AGKW}.  The full result, as stated in Theorem~\ref{thmloc}, with Anderson localization (pure point spectrum plus uniform exponential decay of eigenfunctions) and dynamical localization, is proved in \cite{GKnext} by an extension of the Bourgain-Kenig analysis.

\section{Construction of the probability measure}
\label{sectcantor}

We construct a Cantor set as follows. We fix $\beta>0$ and a initial scale  $L_1>1$, define a sequence of scales $L_{k+1}=L_k^\beta$, $k=1,2,\ldots$, and set $\alpha_k = \exp(-L_k)$. Let $\cK^{(0)}=[0,1]$. We remove the middle part of the interval and keep at the edges two intervals: $I^{(1)}_{0}$ on the left and $I^{(1)}_{1}$ on the right, both of length $\alpha_1$. We then repeat the procedure of each of the intervals $I^{(1)}_{0}$ and $I^{(1)}_{1}$. At generation $k$, $\cK^{(k)}$ is the union of $2^k$ disjoint intervals $I^{(k)}_j$, $j=(j_1,\cdots, j_k)\in\{0,1\}^k$, each of them with length $\alpha_k$, namely
\begin{equation}\label{Ck}
\cK^{(k)} = \bigcup_{j\in \{0,1\}^k} I^{(k)}_j \mbox{ with } |I^{(k)}_j|=\alpha_k = \exp(-L_k).
\end{equation}
Note that two intervals $I^{(k)}_j$ and $I^{(k)}_{j'}$, $j\not= j'$, are separated by a gap of size at least 
\beq\label{defgap}
G_k:= \alpha_{k-1}-2\alpha_{k} \approx \exp(-L_{k-1}).
\eeq
The Lebesgue measure of $\cK^{(k)}$ is given by
\beq
|\cK^{(k)}| = 2^k \alpha_k = 2^k \exp(-L_1^{\beta^k}),
\eeq
which goes to zero as $k$ goes to infinity. We define
\beq
\cK = \bigcap_{k=0}^\infty \cK^{(k)}.
\eeq
The set $\cK$ is a Cantor set of zero Lebesgue measure.

We construct the uniform measure $\mu$ on the Cantor set  $\cK$  as follows. At generation $k$ we consider the uniform measure $\mu^{(k)}$ on $\cK^{(k)}$, i.e.,   $\mu^{(k)}$ is the normalized Lebesgue measure on each interval $I^{(k)}_j$, with $\mu^{(k)}(I^{(k)}_j)=2^{-k}$  with for all $j$.   Note $\mu^{(k)}([0,1])=1$ by construction. The Cantor measure $\mu$ is  the unique weak limit of the $\mu^{(k)}$'s (to see uniqueness, it is enough to compute the limit on arbitrary intervals). The support of $\mu$ is the Cantor set $\cK$, which  has zero Lebesgue measure by construction. The measure $\mu$ is purely singular continuous.

It is actually easy to see that $\mu$ is $\log\log$-H\"older continuous, but not better. Indeed for any $k,j$ we have 
\beq
\mu(I^{(k)}_j)=\mu^{(k)}(I^{(k)}_j)=2^{-k},
\eeq
that is, with $\eps=|I^{(k)}_j|= \exp(-L_1^{\beta^k})$,
$$
\mu(I^{(k)}_j) = \left(\frac{\log|\log\eps|}{\log L_1}\right)^{-\log 2 / \log \beta};
$$
on the other hand, for any interval $I$ of size $\eps>0$, recalling \eqref{defgap}, if $G_{k+1}\le \eps < G_k$, then $I$ covers at most one interval $I^{(k)}_j$, and thus 
$$
\mu(I) \le 2^{-k} \lesssim \left(\frac{\log|\log\eps|}{\log L_1}\right)^{-\log 2 / \log \beta}.
$$

\section{Elements of the Multiscale Analysis}

Finite volume operators are defined as in  \cite{GHK2}.   Given   a box  $\Lambda= \Lambda_{L}(x)$  in $\R^{d}$ and a configuration $\om\in [0,1]^{\Z^d}$, we set 
\begin{align}\label{finvolH}
H_{\om,\Lambda} :=-{\Delta_{\Lambda}}+ V_{\om,\Lambda} \quad \text{on}   \quad \L^{2}(\Lambda),
\end{align}
where $\Delta_{\Lambda}$ is the  Laplacian on $\Lambda$ with Dirichlet boundary condition,  and  
\begin{equation}
 V_{\om,\Lambda}:= \chi_{\Lambda} \sum_{\zeta \in \Z^d\cap \Lambda} \om_\zeta \,u(x - \zeta).
 \end{equation}
  The finite volume resolvent  is 
 $R_{X,\Lambda} (z):=(H_{X,\Lambda} - z)^{-1}$.

 We will  identify $\L^{2}(\Lambda)$ with $\chi_{\Lambda }\L^{2}(\R^{d})$.
Note that in general we do not have $ V_{\om,\Lambda}=  \chi_{\Lambda} V_{\om,{\Lambda^{\prime}}}$ for   $\Lambda \subset \Lambda^{\prime}$, where $\Lambda^{\prime}$ may be a finite box or $\R^{d}$. But we always have
\begin{equation}
 \chi_{\widehat{\Lambda}} V_{\om,\Lambda}=  \chi_{\widehat{\Lambda}} V_{\om,{\Lambda^{\prime}}}, 
\end{equation}
where
\beq    \widehat{\Lambda}=\widehat{\Lambda}_{L}(x):=\Lambda_{L-\delta_{+} }(x) \quad  \text{with $\delta_{+}$ as in \eqref{u}},  \label{lambhat}
\eeq
which suffices for the multiscale analysis.

  The usual  definition of ``good" boxes for the multiscale analysis is as follows.

\begin{definition} \label{defwegner} Let $\eps\in]0,1[$ (small). Consider $\om \in [0,1]^{\Z^d}$, an energy $E\in \R$, and  a rate of decay $m>0$. A  box  $\Lambda_{L}$ is said to be $(\om,E,m)$-good if 
  \begin{align}\label{weg}
\| R_{\om,\Lambda_{L}}(E) \|& \le \e^{L^{1-\eps}}
\intertext{and} 
\| \chi_x R_{\om,\Lambda_{L}}(E) \chi_y \|& \le \e^{-m |x-y|},\quad  \text{for all  $x,y \in \Lambda_{L}$ with $ |x-y|\ge \tfrac L{10}$}. \label{good}
\end{align} 
\end{definition}

But  \emph{goodness} of boxes does not suffice for the induction step in the multiscale analysis given in \cite{B,BK}, which also needs an adequate supply of \emph{free sites}  to obtain a Wegner estimate at each scale.  Given $S \subset{\Lambda}\cap \Z^d$ and $t_S=\{t_\zeta\}_{\zeta \in S} \in [0,1]^S$, we set
\begin{align}\label{finvolHS}
H_{{\om},{t_S},\Lambda} :=-{\Delta_{\Lambda}}+ \chi_\Lambda V_{{\om}_\Lambda,{t_S}} \quad \text{on}   \quad \L^{2}(\Lambda),
\end{align}
where
\begin{equation}
 V_{{\om_\Lambda},{t_S}}(x):= \sum_{\zeta \in\pa{{\Lambda}\cap \Z^d}\setminus S} \omega_\zeta \,u(x - \zeta)+ \sum_{\zeta \in  S} t_\zeta \,u(x - \zeta).  \label{finvolVS}
\end{equation}
 $R_{\om,{t_S},\Lambda} (z)$ will denote the corresponding finite volume resolvent. Following \cite{BK}, sites belonging to $S$ are called {\it free sites}, and variables $t_\zeta$ with $\zeta\in S$ are called {\it free variables}.

\begin{definition} Consider $\om \in [0,1]^{\Z^d}$, an energy $E\in \R$,  a rate of decay $m>0$, and $S \subset \widehat{\Lambda}$.    A  box  $\Lambda_{L}$ is said to be $(\om ,S,E,m)$-good if 
 we have   \eqref{weg} and  \eqref{good} 
with  $R_{\om,{t_S},\Lambda} (E)$ for all $t_S \in [0,1]^{S}$.  In this case $S$ consists of $(\om,E)$-free sites for the box  $\Lambda_{L}$. 
\end{definition}
Recall the sequence of scales $L_k$ from the construction of the Cantor set $\cK$. Let $\Lambda=\Lambda_L$ be a cube of side $L$.

For $\eps>0$ (the one in Definition~\ref{defwegner}), we define $\eps'>0$ such that
\beq\label{eps}
(1+\eps')(1-\eps)=1.
\eeq

\begin{definition}\label{defclass}
Given a cube $\Lambda=\Lambda_L$, two configurations $\omega,\omega'\in\Omega$ are said to be equivalent at scale $L\in]L_{k-1}^{1+\eps'}, L_{k}^{1+\eps'}]$ (thus given $L$, the integer $k$ is unique), if for all $i\in\Z^d\cap\Lambda$, $\omega_i$ and $\omega'_i$ belong to the same interval $I^{(k)}_j$ of  $\cK^{(k)}$, and we write $\omega_i\sim_L\omega'_i$, $i\in\Z^d\cap\Lambda$, and $\omega\sim_{\Lambda}\omega'$ for the full configuration.
We denote by $[\omega_i]_{L}$ the equivalent class of $\omega_i$, and by $[\omega]_{\Lambda}$ that of $\omega$ in $\Lambda$.
\end{definition}

\begin{remark}\label{remclass}
In other terms, given $L$ as in Definition~\ref{defclass}, the quotient of $\cK=\supp\mu$ by the relation of equivalence $\sim_L$ can be identified with the set of the intervals of $\cK^{(k)}$, and hence with $\{1,\cdots, K\}$, where $K=2^k$ is the number of such intervals. We also identify the set of all equivalence classes $[\omega]_{\Lambda}$ with $\{1,\cdots, K\}^{\Z^d\cap\Lambda}$, and by $[\omega_i]_L=l$, $1\le l\le K$, we mean that we identify  $[\omega_i]_L$ with the $l^{\rm th}$ interval of $\cK^{(k)}$.
\end{remark}

Remark~\ref{remclass} motivates the following definition.

\begin{definition}
Let $\Lambda=\Lambda_{L_n}$, $K=2^n$. We denote by $C_\Lambda\simeq \{1,\cdots, K\}^{\Z^d\cap\Lambda}$ the set of collections of intervals $I_j^{(n)}$, indexed by lattice points in $\Lambda$. An element of $C_\Lambda$ will be denoted by $[\Lambda]_\Lambda$ or just $[\Lambda]$.

Similarly, if $A\subset\Lambda\cap \Z^d$, then $C_A\simeq \{1,\cdots, K\}^{A}$ denotes the set collections of intervals $I_j^{(n)}$, indexed by points in $A$. An element of $C_A$ will be denoted by $[A]_\Lambda$ or just $[A]$.

When we want to stress that we work at resolution $k$, we write $C_\Lambda^{(k)}$ and $C_A^{(k)}$.
\end{definition}

We now define the basic events (``bevents") that we shall use for the multiscale analysis. They correspond to \cite{BK}'s cylinders in our particular setting.  We rely on the construction introduced in \cite{GHK2} but the situation is a bit simpler since we do not have to introduce ``acceptable" configurations as in \cite{GHK2}. In some sense, all our configurations are ``acceptable".

\begin{definition} Let us give $\eps>0$ and a box $\Lambda=\Lambda_{L}(x)$, and $B,S$ a partition of $\Lambda\cap\Z^d$. For $k$ s.t. $L\in]L_{k-1}^{1+\eps'}, L_k^{1+\eps'}]$  we write $C_\Lambda$ in place of $C_\Lambda^{(k)}$. We call $k$ the resolution associated to $L$. An element $[\Lambda]$ of $C_\Lambda$ is written  $[\Lambda]=([B] , [S])$. A $\Lambda$-bconfset (basic configuration set) is a subset of $C_\Lambda$ of the form
\beq \label{cylinderset}
C_{\Lambda,[B],S}:=\left\{([B] , [S]), [S]\in C_S \right\},
\eeq
or, stressing we work at resolution $k$ (intervals of length $\e^{-L_k}$),
\beq \label{cylindersetresol}
C^{(k)}_{\Lambda,[B]_{L_k},S}:=\left\{([B]_{L_k} , [S]_{L_k}), [S]_{L_k}\in C_S^{(k)} \right\},
\eeq
where $[B]\in C_B$. $C_{\Lambda,[B],S}$ is a $\Lambda$-dense bconfset if
 the set of indices $S$ satisfies the density condition 
  \begin{equation}  \label{densitycond}
\# (S\cap {\widehat{\Lambda}_{L^{1-}}} )\ge L^{d-} \quad \text{for all
boxes}\quad   \Lambda_{L^{1-}}\subset \Lambda_{L}.
\end{equation}
Note that if $S=\emptyset$ then $B=\Lambda$. We then set
\beq  \label{S=empty}
C_{\Lambda,[B]}:=C_{\Lambda,[B],\emptyset}=[B] .
\eeq
We turn to random events. A $\Lambda$-bevent (basic event) is  a  subset of  $ \Omega$ of the form
\beq \label{bevent}
\cC_{\Lambda,[B],S}:= \{\omega\in\Omega, \; \{\omega\}_{i\in\Lambda} \in  C_{\Lambda,[B],S}\},
\eeq
where $[B]\in C_B$, with the analog of \eqref{cylindersetresol} at resolution $k$. $\cC_{\Lambda,[B],S}$ is a $\Lambda$-dense bevent if
  $S$ satisfies the density condition \eqref{densitycond}. In addition, we set
\beq  \label{S=empty2}
\cC_{\Lambda,[B]}:=\cC_{\Lambda,[B],\emptyset}=\{\omega\in\Omega, \; \{\omega\}_{i\in\Lambda} \in C_{\Lambda,[B]}\} = \{\omega\in\Omega, \; \{\omega\}_{i\in \Lambda} =[B]\}.
\eeq
\end{definition}

Note that for each $S_{1}\subset S$ we  have 
\begin{align} \label{cylinderexp}
C_{\Lambda,[B],S}&= \bigsqcup_{[S_1]\in C_{S_1}} C_{\Lambda,[B] \sqcup [S_1],S\setminus S_{1}},\\
\cC_{\Lambda,[B],S}&= \bigsqcup_{[S_1]\in C_{S_1}} \cC_{\Lambda,[B] \sqcup [S_1],S\setminus S_{1}}, \label{cylinderexp2}
\end{align}
where $\sqcup$ denotes a disjoint union.

When changing scales, one redraws  cylinders in the most natural way:
\beq
C^{(k)}_{\Lambda_{L_k},[B]_{L_k},S} = \bigsqcup_{[B]_{L_{k+1}}\subset[B]_{L_k}} C^{(k+1)}_{\Lambda_{L_{k}},[B']_{L_{k+1}},S} \, ,
\eeq
where $[B]_{L_{k+1}}= \set{I^{(k+1)}_{j_i}}_i\subset[B]_{L_k}=\set{I^{(k)}_{j_i}}_i$ if and only if $I^{(k+1)}_{j_i}\subset I^{(k)}_{j_i}$ for all $i\in\Lambda_{k+1}$.

\begin{definition} \label{defadapted}
 Consider an energy $E \in \R$, $m>0$, and a box 
  $\Lambda=\Lambda_{L}(x)$.  
The $\Lambda$-bevent $\cC_{\Lambda,[B],S}$ and the $\Lambda$-bconfset $C_{\Lambda,[B],S}$  are  $(\Lambda,E,m)$-good if the
box $\Lambda$ is $([B],S,E,m)$-good. 
\end{definition} 

\begin{definition}  Consider an energy $E \in \R$, a rate of decay  $m>0$, and a box   $\Lambda$.  We call $\Omega_{\Lambda}$ a  $(\Lambda,E,m)$-localized event if
there exist  disjoint  $(\Lambda,E,m)$-good dense  bevents
	$\{\cC_{\Lambda,[B_{i}],S_{i}}\}_{i=1,2,\ldots,I}$  such that 
\beq \label{adapted3}
\Omega_{\Lambda}= \bigsqcup_{i=1}^{I} \cC_{\Lambda,[B_{i}],S_{i}} .
\eeq
\end{definition}

We prove the following multiscale analysis.

\begin{proposition}\label{propMSA}
Let $\eps>0$, $\beta=4/3+$ and $p<\frac38$ be given. There exists $E_0>0$ and $m>0$ such that if $L_1$ is large enough, then for all $E\in [0,E_0]$ and all $L\ge L_1$, there exists an $(\Lambda_L,E,m)$-localized event $\Omega_{\Lambda_L}$ with $\P(\Omega_{\Lambda_L})\ge 1 - L^{-pd}$.
\end{proposition}

To prove Proposition~\ref{propMSA}, we use the Bourgain-Kenig multiscale analysis adapted to bevents as done in \cite{GHK2}. The Wegner estimate of Bourgain-Kenig as stated in \cite[Lemma~5.1 (and 5.1')]{BK} is translated into ``bevents" language in  \cite[Lemma~5.10]{GHK2}.

Theorem~\ref{thmloc} follows from Proposition~\ref{propMSA} as in  \cite[Section~6]{GHK2}.
The initial condition can be obtained by  the argument in \cite{BK,GHK2}  using a large deviation result for $\mu$ to estimate the probability that the averaged sum of the random variables in a cube of size, say, $(\log L_0)^2$ is less than half its mean, where $L_0$ is the initial length scale. We refer to \cite{GKnext} for a proof of the initial condition with an arbitrary measure $\mu$, as well as for dynamical localization.

It thus remains to prove  \cite[Lemma~5.10]{GHK2} for our particular probability distribution $\mu$. This is the purpose of Section~\ref{sectWegner}.

%%%%%%%%%%%%%%%%%%%%%%%%%%%%%%%%%%%%%%%%%%%%%%%%%%%%%%%%%%%%%%%%%%
%%%%%%%%%%%%%%%%%%%%%%%%%%%%%%%%%%%%%%%%%%%%%%%%%%%%%%%%%%%%%%%%%%%%%%%%%%%%%%%%%%%%%%

\section{Maximal antichain in posets}

In this section we briefly collect some tools and facts coming from the theory of posets (partially ordered multisets).

Let $\cM=\{1,2,\cdots,K\}^n$ be a multiset with partial order $x\le y$ iff $x_i\le y_i$ for all $1\le i \le n$. Two elements $x,y\in\cA$ are are comparable if $x\le y$ or $y\le x$. We define the rank function $r(x)=\sum x_i$ and the rank number $N_r$  as the number of $x\in\cM$ with rank $r$. An antichain $\cA\subset \cM$ is a collection of $x$'s such that no elements $x,y\in\cA$ are comparable. We further equip $\cM$ with the  discrete uniform probability structure: $\P_{\cM}(x)=K^{-n}$.

\begin{lemma}\label{lemsperner}
Let $\cA\subset \cM$ an antichain, then $\P_{\cM}(\cA)\le c(K \sqrt{n})^{-1}$.
\end{lemma}

\begin{remark}\label{remK}
Note that in our particular setting we will have $K=2^{k+1}\approx (\log n)^{\log 2 / \log \beta}$, so the factor $K^{-1}$ in \eqref{discprob} does not really improve the probability.
\end{remark}

\begin{proof}
We first recall the LYM inequality \cite[Theorem~2.3.1]{A}: It asserts that if $\cA$ is an antichain then
\begin{equation}\label{LYM}
\sum_{x\in\cA} \frac1{N_{r(x)}} \le 1.
\end{equation}
We combine \eqref{LYM} with an estimate on the maximal rank number \cite[Theorem~4.3.6]{A}, namely, for some positive constants $c,C$,
\begin{equation}
\frac{cK^n}{\sqrt{(K-1)(K+1)}\sqrt{n}} \le \max_{r} N_{r} \le \frac{CK^n}{\sqrt{(K-1)(K+1)}\sqrt{n}} .
\end{equation}
  We thus have 
\beq\label{discprob}
\P_{\cM}(\cA)= \frac1{K^n} \sum_{x\in \cA} 1 \le \frac{\max N_r} {K^n} \sum_{x\in \cA} \frac1{N_{r(x)}} \le \frac{\max N_{r}} {K^n} \lesssim \frac1{K \sqrt{n}}.
\eeq
\end{proof}

%%%%%%%%%%%%%%%%%%%%%%%%%%%%%%%%%%%%%%%%%%%%%%
%%%%%%%%%%%%%%%%%%%%%%%%%%%%%%%%%%%%%%%%%%%%%%
%%%%%%%%%%%%%%%%%%%%%%%%%%%%%%%%%%%%%%%%%%%%%%
%%%%%%%%%%%%%%%%%%%%%%%%%%%%%%%%%%%%%%%%%%%%%%

\section{The Wegner estimate}
\label{sectWegner}

To prove \cite[Lemma~5.10]{GHK2} in our particular setting we need to prove the following lemma.

\begin{lemma} Consider a box $\Lambda=\La_L$, let $\ell=L^\rho$ with $\rho= \beta^{-1}=\frac 3 4 -$
 (so $L^{\frac 4 3 \rho} =L^{1-}$).  Let $\cC_{\Lambda,[B],S}\subset\Omega$ be given. Let $S \subset {\widehat{\Lambda}}$ with $\abs{S}= \ell^{d-}$,  pick (and fix) $\omega \in \cC_{\Lambda,[B],S}$, and
 set
 \beq
 H(t_S):=H_{{\om},{t_S},\Lambda}  \quad \text{for all $t_S \in  [0,1]^S$}.
 \eeq
Consider an energy $E_0$, set $I= (E_0 - \e^{-c_1 \ell}, E_0 +\e^{-c_1 \ell})$.
 Let $E_\tau(\om,t_S)$ be a continuous eigenvalue parametrization of $\sigma(H(t_S))$
 such that $E_\tau(0) \in I$ (a finite family).  Let $E_\om(t_S)=E_{\tau_0}(t_S)$ for some $\tau_0$, and $E(\om)=E_\om(\om_S)$, $\om_S=\{\omega\}_{i\in S}$.   Suppose
 \beq\label{hypderiv}
\e^{-c_3 \ell^{\frac 4 3}  \log \ell}\le \frac {\partial}  {\partial t_i }E(t_S)\le \e^{- c_2 \ell} \quad \text{for all $i \in S$ if $E(t_S) \in I$}.
 \eeq
Then, there exists $\Omega_{[B]}=\bigsqcup_{l} ([B],[S]_l)\subset \cC_{\Lambda,[B],S}$ (the disjoint union being finite), such that
 \beq\label{wegset}
\set{\om \in \cC_{\Lambda,[B],S}; \; E(\omega )\in (E_0 -\e^{-2 c_3 \ell^{\frac 4 3}  \log \ell}, E_0 +\e^{-2c_3 \ell^{\frac 4 3}  \log \ell})  } \subset \cC_{\Lambda,[B],S}\setminus \Omega_{[B]},
 \eeq
 and
 \beq\label{wegproba}
 \P (\cC_{\Lambda,[B],S}\setminus \Omega_{[B]}; \; \cC_{\Lambda,[B],S}) \le  C \ell^{- \frac  d 2 +}.
 \eeq
 \end{lemma}

\begin{proof}
Recall \eqref{eps}, that is we are given $\eps'>0$ such that $(1+\eps')(1-\eps)=1$. We have $L\in[L_k^{1+\eps'}, L_{k+1}^{1+\eps'}]$ and thus 
\beq\label{Leps}
L_k \le L^{1-\eps}\le L_{k+1}.
\eeq
We set $\ell=L^\rho$ and assume that $\Lambda_\ell$ is a good box with probability $\ge 1 - \ell^{-p}$, as well as at previous scales.
Let $\omega\in\cC_{\Lambda,[B],S}$ and assume that 
\beq
E(\omega )\in (E_0 -\e^{- 2c_3 \ell^{\frac 4 3}  \log \ell}, E_0 +\e^{-2c_3 \ell^{\frac 4 3}  \log \ell}). 
\eeq
We show that for any $\om'$ such that $[\om]=([B],[S])$ and $[\om']=([B],[S]')$ are comparable and distinct, one has
\beq\label{intervbis}
E(\om')\not\in (E_0 -\e^{-2 c_3 \ell^{\frac 4 3}  \log \ell}, E_0 +\e^{-2c_3 \ell^{\frac 4 3}  \log \ell}).
\eeq
Indeed, suppose, for instance, that $[\om]<[\om']$. If $[\om_i]=[\om'_i]$ then $|\om'_l - \om_l|\le \e^{-L_{k+1}}$, due to  Definition~\ref{defclass}. And if $[\om_i]<[\om'_i]$ then by construction these two points are separated by a gap of length $\ge \frac12 \e^{-L_k}$ (for $L$ large enough). Thanks to \eqref{hypderiv}, we get
\begin{align}
|E(\om')-E(\om)| & \ge  \frac12 \e^{- c_3 \ell^{\frac 4 3}  \log \ell} \e^{-L_k} - \ell^{d-} \e^{- c_2 \ell} \e^{-L_{k+1}} \\
& \ge \frac12 \e^{- c_3 \ell^{\frac 4 3}  \log \ell} \e^{-L_k}
\left(1 -  4\e^{-L_{k+1}+L_k}\e^{ c_3 \ell^{\frac 4 3}  \log \ell} \right) \\
& \ge \frac12 \e^{- c_3 \ell^{\frac 4 3}  \log \ell} \e^{-L_k}
\left(1 -  4\e^{-\frac14 L_{k+1}} \right) \\
& \ge \e^{- c_3 \ell^{\frac 4 3}  \log \ell} 
\end{align}
for $L$ large enough, where we used that
\beq
c_3 \ell^{\frac 4 3}  \log \ell \le \frac14 L^{1-\eps}.
\eeq
We thus obtain \eqref{intervbis}. As a consequence $E(\om')\in (E_0 -\e^{- 2c_3 \ell^{\frac 4 3}  \log \ell}, E_0 +\e^{-2c_3 \ell^{\frac 4 3}  \log \ell})$ can only happen for a set of $\omega'$ such that the associated $[S]$ and $[S]'$ are non comparable, that is, for a collection of $[S]$'s that belongs to a finite antichain $\cA_{[B]}=\set{[S]_1, [S]_2, \cdots}$. Setting 
\beq
\cC_{\Lambda,[B],S}\setminus\Omega_{[B]}:= ([B],\cA_{[B]}) = \bigsqcup_l ([B],[S]_l),
\eeq
we have proved \eqref{wegset}. It remains to show the probabilistic estimate \eqref{wegproba}.

Let us give an element $[S]=\set{I^{(k+1)}_{j_i}}_{i\in S}\in C_S$  and $\om\in\cC_{\Lambda,[B],S}$. By construction of $\mu$, $\mu(I^{(k+1)}_{j_i})=K^{-1}$ for any $i\in S$, where $K= 2^{k+1}$. We thus replace the (continuous) probability space $(\Omega,\P)$ conditionned to $\cC_{\Lambda,[B],S}$ by the discrete probability space $(C_S, \P^o_S)\simeq(\{[S]\in C_S\}, \P^o_S) $, where $\P^o_S$ is the uniform measure with  weight $K^{-|S|}$. 
We thus get 
\begin{align}
& \P\set{\cC_{\Lambda,[B],S}\setminus\Omega_{[B]}; \; \cC_{\Lambda,[B],S}} =  \P^o_S(\cA_{[B]}).
\end{align}

We now apply Lemma~\ref{lemsperner} to the set $C_S\simeq\{1,2,\cdots,K\}^n$, with $K=2^{k+1}$ and $n=|S|$, equipped with the discrete uniform probability $\P^o_S$. This ends the proof.
\end{proof}

%%%%%%%%%%%%%%%%%%%%%%%%%%%%%%%%%%%%%%%%%%%%%%%%%%%%%%%%%%%%%%%%%%
%%%%%%%%%%%%%%%%%%%%%%%%%%%%%%%%%%%%%%%%%%%%%%%%%%%%%%%%%%%%%%%%%%%%%%%%%%%%%%%%%%%%%%%%%%%%%%%%%%%%%%%%%%%%%%%%%%%%%%%%%%%%%%%%%%%%

%%%%%%%%%%%%%%%%%%%%%%%%%%%%%%%%%%

\end{document}